\documentclass[a4paper,UKenglish]{lipics-v2018}

\hideLIPIcs

\usepackage{microtype}
\usepackage{xspace}
\usepackage{algorithm}
\usepackage[noend]{algorithmic}

\graphicspath{{figures/}}

\newcommand{\etal}{et~al.\xspace}

\title{Optimal Morphs of Planar Orthogonal Drawings}

\author{A. van Goethem}{TU Eindhoven\\{Eindhoven, the Netherlands}}{a.i.v.goethem@tue.nl}{}{}

\author{K. Verbeek}{TU Eindhoven\\{Eindhoven, the Netherlands}}{k.a.b.verbeek@tue.nl}{}{[K. Verbeek
is supported by the Netherlands Organisation for Scientific Research (NWO) under
project no. 639.021.541.]}

\authorrunning{A. van Goethem and K. Verbeek} 

\Copyright{Arthur I. van Goethem, Kevin A. B. Verbeek}

\subjclass{Theory of computation $\rightarrow$ Graph algorithms analysis}


\relatedversion{A full version of this paper is available at \url{https://arxiv.org/abs/1801.02455}}

\acknowledgements{We want to thank the anonymous reviewers for their extensive feedback.}

\EventEditors{Bettina Speckmann and Csaba D. T{\'o}th}
\EventNoEds{2}
\EventLongTitle{34th International Symposium on Computational Geometry (SoCG 2018)}
\EventShortTitle{SoCG 2018}
\EventAcronym{SoCG}
\EventYear{2018}
\EventDate{June 11--14, 2018}
\EventLocation{Budapest, Hungary}
\EventLogo{socg-logo}
\SeriesVolume{99}
\ArticleNo{42}

\nolinenumbers

\begin{document}

\maketitle
\begin{abstract}
We describe an algorithm that morphs between two planar orthogonal drawings $\Gamma_I$ and $\Gamma_O$ of a connected graph $G$, while preserving planarity and orthogonality. Necessarily $\Gamma_I$ and $\Gamma_O$ share the same combinatorial embedding. Our morph uses a linear number of linear morphs (linear interpolations between two drawings) and preserves linear complexity throughout the process, thereby answering an open question from Biedl~\etal~\cite{Biedl2013}.

Our algorithm first \emph{unifies} the two drawings to ensure an equal number of (virtual) \emph{bends} on each edge. We then interpret bends as vertices which form obstacles for so-called \emph{wires}: horizontal and vertical lines separating the vertices of $\Gamma_O$. We can find corresponding wires in $\Gamma_I$ that share topological properties with the wires in $\Gamma_O$. The structural difference between the two drawings can be captured by the \emph{spirality} of the wires in $\Gamma_I$, which guides our morph from $\Gamma_I$ to $\Gamma_O$.
\end{abstract}

\section{Introduction}
A morph is a continuous transformation between two objects. It is most effective, from a visual point of view, if the number of steps during the transformation is small, if no step is overly complex, and if the morphing object retains some similarity to input and output throughout the process. These visual requirements can be translated to a variety of algorithmic requirements that depend on the type of object to be morphed.

In this paper we focus on morphs between two planar orthogonal drawings of a connected graph $G$ with complexity $n$. In this setting the visual requirements for a good morph can be captured as follows: few (ideally at most linearly many) steps in the morph, each step is a simple (ideally linear) morph, and each intermittent drawing is a planar orthogonal drawing of $G$ with complexity $O(n)$. Biedl~\etal~\cite{Biedl2006} presented some of the first results on this topic, for the special case of \emph{parallel} drawings: two graph drawings are parallel when every edge has the same orientation in both drawings. The authors proved that there exists a morph, which is composed of $O(n)$ linear morphs, between two parallel drawings that maintains parallelity and planarity for orthogonal drawings.
More recently, Biedl~\etal~\cite{Biedl2013} described a morph, composed of $O(n^2)$ linear morphs, between two planar orthogonal drawings that preserves planarity, orthogonality, and linear complexity during the morph.
The authors also present a lower bound example requiring $\Omega(n^2)$ linear morphs when morphed with their method.

In this paper we present a significant improvement upon their work: we describe a morph, which is composed of $O(n)$ linear morphs, between two planar orthogonal drawings which preserves planarity, orthogonality, and linear complexity during the morph. This bound is tight as directly follows from the general lowerbound for straight-line graphs proven in~\cite{Alamdari2016}.

\subparagraph*{Related work} Morphs of planar graph drawings have been studied extensively, below we review some of the most relevant results. Cairns showed already in 1944~\cite{Cairns1944} that there exists a planarity-preserving continuous morph between any two (compatible) triangulations that have the same outer triangle. His proof is constructive and results in an exponential time algorithm to find such a morph. These results were extended in 1983 by Thomassen~\cite{Thomassen1983} who showed that two compatible straight-line drawings can be morphed into each other while maintaining planarity (still using exponential time). Thomassen also proved that two rectilinear polygons with the same turn sequence can be transformed into each other using a sequence of linear morphs. Much more recently, Angelini~\etal~\cite{Angelini2013} proved that there is a morph between any pair of planar straight-line drawings of the same graph (with the same embedding) using $O(n^2)$ linear morphs. Finally, Alamdari~\etal~\cite{Alamdari2016} improved this result to $O(n)$ uni-directional linear morphs, which is optimal. Creating this morph takes $O(n^3)$ time. It does create intermediate drawings which need to be represented by a superlogarithmic number of bits, leaving as a final open question if it is possible to morph two planar straight-line drawings using a linear number of linear morphs while using a logarithmic number of bits per coordinate to represent intermediate drawings. Note that, since intermediate drawings are not orthogonal, we cannot apply this approach to our setting. Our approach relies heavily on the spirality of the drawings. This concept of spirality has already received significant attention in the area of bend-minimization (e.g.,~\cite{Blasius2016,DiBattista1998,Didimo2009}).

\subparagraph*{Paper outline} Our input consists of two planar orthogonal drawings $\Gamma_I$ and $\Gamma_O$ of a connected graph $G$, which share the same combinatorial embedding. In Section~\ref{sec:prelim} we first give all necessary definitions and then explain how to create a \emph{unified} graph $G$: we add ``virtual'' bends to edges to ensure that each edge is drawn in $\Gamma_I$ and $\Gamma_O$ with the same number of bends. We then interpret each bend as a vertex of the unified graph $G$.
$\Gamma_I$ and $\Gamma_O$ are now orthogonal straight-line drawings of the unified graph $G$.
Clearly the maximum complexity of $\Gamma_I$ and $\Gamma_O$ is still bounded by $O(n)$ after the unification process.

Our main tool are so-called \emph{wires} which are introduced in Section~\ref{sec:wires}. Wires capture the horizontal and vertical order of the vertices. Specifically, we consider a set of horizontal and vertical lines that separate the vertices of $\Gamma_O$. If we consider the vertices of $\Gamma_O$ as obstacles, then these wires define homotopy classes with respect to the vertices of $G$ (for the combinatorial embedding of $G$ shared by $\Gamma_I$ and $\Gamma_O$). These homotopy classes can be represented by orthogonal polylines (also called wires) in $\Gamma_I$ using orthogonal shortest and lowest paths as defined by Speckmann and Verbeek~\cite{Speckmann2018}. A theorem by Freedman, Hass, and Scott~\cite{Freedman1982} proves that the resulting paths minimize crossings.

Intuitively our morph is simply straightening the wires in $\Gamma_I$ using the \emph{spirality} (the difference between the number of left and right turns) of the wires as a guiding principle.
In Section~\ref{sec:linear} we show how this approach leads more or less directly to a linear number of linear morphs. However, the complexity of the intermediate drawings created by this algorithm might increase to $\Theta(n^2)$.
In Section~\ref{sec:linearComplexity} we show how to refine our approach, to arrive at a linear number of linear morphs which preserve linear complexity of the intermediate drawings.

\subparagraph*{Relation to Biedl~\etal~\cite{Biedl2013}}
While the underlying principle of our algorithm is quite different, there are certain similarities between our approach and the one employed by Biedl~\etal~\cite{Biedl2013}. As mentioned, there is a lower bound that proves their method cannot do better than $O(n^2)$ linear morphs in general. We sketch why this lower bound does not apply to our algorithm.

\begin{figure}[t]
  \centering
  \includegraphics[]{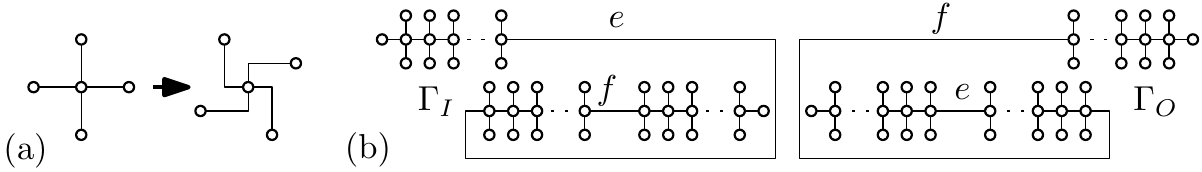}
  \caption{\textbf{(a)} A vertex twist (from~\cite{Biedl2013}). \textbf{(b)} The input and output for the lowerbound example from~\cite{Biedl2013}. There are three parts each containing a linear number of vertices. Edge $e$ in $\Gamma_I$, respectively $f$ in $\Gamma_O$, has a linear number of bends (only four drawn in example).}
  \label{fig:twist}
\end{figure}

To ensure that the spirality of all edges is the same in both the input and the output, the algorithm by Biedl~\etal ``twists'' the vertices (see Fig.~\ref{fig:twist}(a)). The lowerbound described in~\cite{Biedl2013} (see Fig.~\ref{fig:twist}(b)) shows that it may be necessary to twist a linear number of vertices a linear number of times. The complexity introduced in the edges causes another quadratic number of linear moves to keep the complexity of the drawing low.

Our approach must overcome the same problem: a linear number of vertices (edges) might need to be rotated a linear number of times. The crucial difference is that our algorithm can rotate a linear number of vertices (edges) at once, using only $O(1)$ linear morphs. We do not require the edges to have the correct spirality at the start of the morph.
Instead we combine the twisting (rotating) of the vertices with linear moves on the edges and pick a suitable order for rotations based on the spirality of the complete drawing. 
As a result we can change the spirality of a linear number of edges in $\Theta(1)$ linear morphs, and we can  rotate a linear number of vertices in $\Theta(1)$ linear morphs.
In the full version of the paper we show how our algorithm works on the lowerbound example of~\cite{Biedl2013}.

\section{Preliminaries}\label{sec:prelim}
\subparagraph*{Orthogonal drawings}
A \emph{drawing} $\Gamma$ of a graph $G = (V, E)$ is a mapping of each vertex to a distinct point in the plane and each edge $(u, v)$ to a curve in the plane connecting $\Gamma(u)$ and $\Gamma(v)$. A drawing is \emph{orthogonal} if each curve representing an edge is an orthogonal polyline consisting of horizontal and vertical segments, and a drawing is \emph{planar} if no two curves representing edges intersect in an internal point. Two drawings $\Gamma$ and $\Gamma'$ of the same graph $G$ have the same \emph{combinatorial embedding} if at every vertex of $G$ the cyclic order of incident edges is the same in both $\Gamma$ and $\Gamma'$.

Let $\Gamma$ and $\Gamma'$ be two planar drawings with the same combinatorial embedding. A \emph{linear morph} $\Gamma_t$ ($0 \leq t \leq 1$) from $\Gamma$ to $\Gamma'$ consists of a linear interpolation between the two drawings $\Gamma$ and $\Gamma'$, that is, $\Gamma_0 = \Gamma$, $\Gamma_1 = \Gamma'$, and for each vertex $v$, $\Gamma_t(v) = (1-t) \Gamma(v) + t \Gamma'(v)$. A linear morph \emph{maintains planarity} if all intermediate drawings $\Gamma_t$ are also planar. Note that a linear morph from $\Gamma$ to $\Gamma'$ may not maintain planarity even if $\Gamma$ and $\Gamma'$ are planar, and that the linear morph may maintain planarity only if $\Gamma$ and $\Gamma'$ have the same combinatorial embedding.
Therefore, a morph between two planar drawings that maintains planarity generally has to be composed of several linear morphs.

\subparagraph*{Slides}
Following Biedl~\etal~\cite{Biedl2013} we generally use two types of linear morphs: zigzag-eliminating slides and bend-creating slides.
Let a \emph{bend} be the shared endpoint of two consecutive segments of an edge.
A \emph{zigzag} consists of three segments joined by two consecutive bends $\beta,\gamma$ that form a left followed by a right bend, or vice versa.
We call zigzags starting and ending with a horizontal segment \emph{horizontal zigzags} (see Fig.~\ref{fig:slides}(a)), and the rest \emph{vertical zigzags}.
We consider the situation with bends $\beta,\gamma$ of Fig.~\ref{fig:slides}(a), other situations are symmetric. Let $\mathcal{V}$ be the set of vertices and bends that are strictly left of $\beta$ and above or at the same height as $\beta$, or strictly above $\gamma$. Also include $\beta$ in $\mathcal{V}$. 
A \emph{zigzag-eliminating slide} moves all points in $\mathcal{V}$ up by the initial distance between $\beta$ and $\gamma$ (see Fig.~\ref{fig:slides}(b)). 
A zigzag-eliminating slide is a linear morph and it straightens the zigzag to a single horizontal or vertical line. The morph always maintains planarity between two drawings.

Inversely, a \emph{bend-creating slide} is a morph that introduces a zigzag in a horizontal or vertical line (see Fig.~\ref{fig:slides}(c)). It can be perceived as the inverse operation of a zigzag-eliminating slide and for similar reasoning is a linear morph that maintains planarity.

\begin{figure}[t]
  \centering
  \includegraphics[]{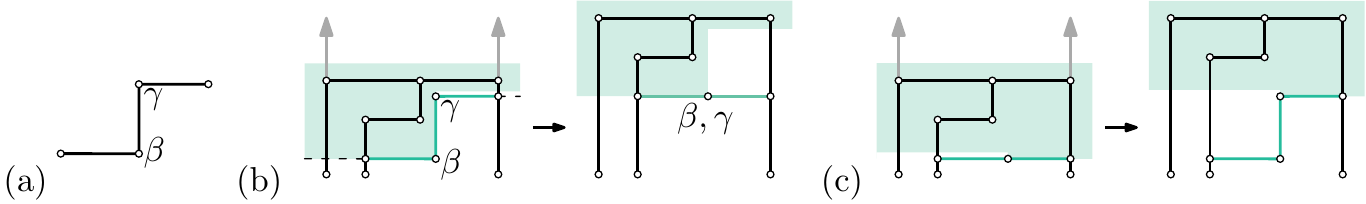}
  \caption{\textbf{(a)} A horizontal \emph{zigzag}. \textbf{(b)} A \emph{zigzag-eliminating} slide is a linear morph straightening a zigzag. \textbf{(c)} A \emph{bend-creating} slide is a linear morph that introduces a zigzag.}
  \label{fig:slides}
\end{figure}

\subparagraph*{Homotopic paths}
Our morphing algorithm heavily relies on the concept of \emph{wires} among the vertices of the drawings, and wires are linked up between different drawings via their homotopy classes. We consider the vertices of a drawing as the set of obstacles $B$. Let $\pi_1, \pi_2\colon [0,1] \rightarrow \mathbb{R}^2 \setminus B$ be two paths in the plane avoiding the vertices. We say that $\pi_1$ and $\pi_2$ are \emph{homotopic} (notation $\pi_1 \sim_h \pi_2$) if they have the same endpoints and there exists a continuous function avoiding $B$ that deforms $\pi_1$ into $\pi_2$. More specifically, there exists a function $\Pi:[0,1] \times [0,1] \rightarrow \mathbb{R}^2$ such that
\begin{itemize}
\item $\Pi(0,t) = \pi_1(t)$ and $\Pi(1,t) = \pi_2(t)$ for all $0\leq t\leq 1$.
\item $\Pi(s,0) = \pi_1(0) = \pi_2(0)$ and $\Pi(s,1) = \pi_1(1) = \pi_2(1)$ for all $0\leq s\leq 1$.
\item $\Pi(\lambda,t) \notin B$ for all $0 \leq \lambda \leq 1$, $0 \leq t \leq 1$.
\end{itemize}
Since the homotopic relation is an equivalence relation, every path belongs to a \emph{homotopy class}. The \emph{geometric intersection number} of a pair of paths $\pi_1,\pi_2$ is the minimum number of intersections between any pair of paths homotopic to $\pi_1$, respectively $\pi_2$. Freedman, Hass, and Scott proved the following theorem\footnote{Reformulated (and simplified) to suit our notation rather than the more involved notation in~\cite{Freedman1982}.}.
\begin{theorem}[from~\cite{Freedman1982}]\label{thm:shortestbest}
Let $M^2$ be a closed, Riemannian 2-manifold, and let $\sigma_1 \subset M^2$ and $\sigma_2 \subset M^2$ be two shortest loops of their respective homotopy classes. If $\pi_1 \sim_h \sigma_1$ and $\pi_2 \sim_h \sigma_2$, then the number of crossings between $\sigma_1$ and $\sigma_2$ is at most the number of crossings between $\pi_1$ and $\pi_2$.
\end{theorem}
In other words, the number of crossings between two loops of fixed homotopy classes are minimized by the shortest respective loops. This theorem can easily be extended to paths instead of loops, if we can consider the endpoints of the paths as obstacles. For orthogonal paths, the shortest path is not uniquely defined and the theorem cannot directly be applied. However, using \emph{lowest paths} the theorem still holds. Refer to~\cite{Speckmann2018} (Lemma 6) for details.

\subparagraph*{Conventions}
When morphing from a drawing $\Gamma$ to a drawing $\Gamma'$, the complexities (number of vertices and bends) of the two drawings may not be the same, as there is no restriction on the complexity of the orthogonal polylines representing the edges.
To simplify the discussion of our algorithm, we first ensure that every two orthogonal polylines in $\Gamma$ and $\Gamma'$ representing the same edge have the same number of segments. This can easily be achieved by subdividing segments, creating additional virtual bends. Next, we eliminate all bends by replacing them with vertices. As a result, all edges of the graph are represented by straight segments (horizontal or vertical) in both $\Gamma$ and $\Gamma'$, and there are no bends. We call the resulting graph the \emph{unification} of $\Gamma$ and $\Gamma'$. If the maximal complexity of $\Gamma$ and $\Gamma'$ is $O(n)$ then clearly the complexity of the unification of $\Gamma$ and $\Gamma'$ is~$O(n)$.

We say that two planar drawings $\Gamma$ and $\Gamma'$ of a unified graph are \emph{similar} if the horizontal and vertical order of the vertices is the same in both drawings. A planar drawing can be morphed to a similar planar drawing using a single linear morph while maintaining planarity.
We can only introduce a crossing if two vertices swap order in the horizontal or vertical direction, which cannot happen during a linear morph between two similar drawings.

Finally, when morphing between two planar drawings $\Gamma$ and $\Gamma'$ of a graph $G$, we assume that $\Gamma$ and $\Gamma'$ have the same combinatorial embedding and the same outer boundary. Furthermore, we assume that $G$ is connected. If $G$ is not connected, then we can use the result by Aloupis~\etal~\cite{Aloupis2015} to connect $G$ in a way that is compatible with both $\Gamma$ and $\Gamma'$. Doing so might increase the complexities of the drawings to $O(n^{1.5})$.

\section{Wires}\label{sec:wires}

In the following we assume that we want to morph an orthogonal planar drawing $\Gamma_I$ of $G=(V,E)$ to another orthogonal planar drawing $\Gamma_O$ of $G$ while maintaining planarity and orthogonality. We further assume that $\Gamma_I$ and $\Gamma_O$ have the same combinatorial embedding and the same outer boundary. We also assume that $G$ is connected, $G$ is the unification of $\Gamma_I$ and $\Gamma_O$, and that $G$ contains $n$ vertices.

To morph $\Gamma_I$ to $\Gamma_O$, our main strategy is to first make $\Gamma_I$ similar to $\Gamma_O$, after which we can morph $\Gamma_I$ to $\Gamma_O$ using a single linear morph. To capture the horizontal and vertical order of the vertices, we use two sets of \emph{wires}. The \emph{lr-wires} $W_{\rightarrow}$, going from left to right through the drawings, capture the vertical order of the vertices. The \emph{tb-wires} $W_{\downarrow}$, going from top to bottom through the drawings, capture the horizontal order of the vertices.

\begin{figure}[t]
  \centering
  \includegraphics[]{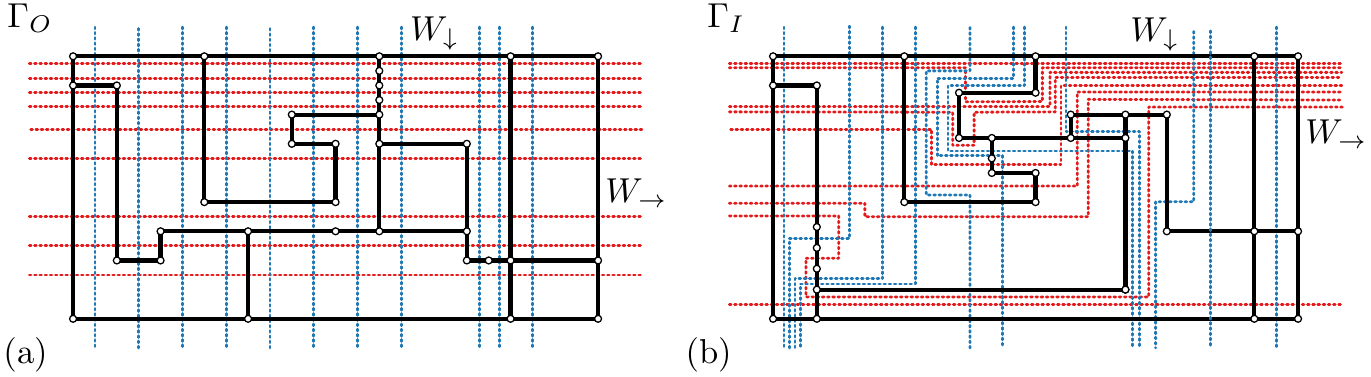}
  \caption{\textbf{(a)} The set $W_\rightarrow$ of lr-wires (red) and the set $W_\downarrow$ of tb-wires (blue) in the output drawing $\Gamma_O$. \textbf{(b)} The matching wires in the input drawing $\Gamma_I$.}
  \label{fig:endGrid}
\end{figure}

Since we want to match the horizontal and vertical order of vertices in $\Gamma_O$, the wires $W_{\rightarrow}$ and $W_{\downarrow}$ are simply horizontal and vertical lines in $\Gamma_O$, respectively, separating any two consecutive coordinates used by vertices (see Fig.~\ref{fig:endGrid}(a)).
Now assume that we have a planar morph from $\Gamma_I$ to $\Gamma_O$ (the existence of such a morph follows from~\cite{Biedl2013}).
If we were to apply this morph in the reverse direction on the wires in $\Gamma_O$, we end up with another set of wires in $\Gamma_I$ with the following properties (see Fig.~\ref{fig:endGrid}(b)): (1) the order of the wires in $W_{\rightarrow}$ ($W_{\downarrow}$) is the same as in $\Gamma_O$ and the same vertices are between consecutive wires, (2) two wires are non-crossing if they both belong to $W_{\rightarrow}$ or $W_{\downarrow}$ and cross exactly once otherwise, and (3) the wires cross exactly the same sequence of edges as in $\Gamma_O$.
These properties follow directly from the fact that a planar morph cannot introduce or remove any crossings, and thus these properties are invariant under planar morphs.
We say that a set of wires is \emph{proper} if it has the above properties.
Interestingly, any proper set of wires can be used to construct a planar morph from $\Gamma_I$ to $\Gamma_O$. We first use a planar morph to straighten the wires.
Then, by Property (1), the resulting drawing $\Gamma$ is similar to $\Gamma_O$, except that it may have some extra bends.
However, by Property (2) the wires form a grid where each cell contains at most one vertex, and the edges crossing the wires are correct by Property (3).
Hence, we can eliminate all bends in a single morph by combining individual morphs per cell.
For each cell we morph all bends (and the vertex) to the center of the cell.
The resulting drawing is similar to $\Gamma_O$ and has no bends, and thus we can finish the planar morph with a single linear morph. 

In the following we assume that we are given a proper set of wires in $\Gamma_I$. Our first goal is to straighten these wires. To keep the distinction between wires and edges clear, we refer to the horizontal and vertical segments of wires as \emph{links}. Note that even a single wire in $\Gamma_I$ may have $\Omega(n^2)$ links (see Fig.~\ref{fig:n2complexity}), so it is not efficient to straighten the wires one link at a time. To straighten the wires more efficiently, we consider the spirality of the wires. For a wire $w \in W_\rightarrow$, let $\ell_1 \ldots \ell_k$ be the links of $w$ in order from left to right. Furthermore, let $b_i$ be the orientation of the bend between $\ell_i$ and $\ell_{i+1}$, where $b_i = 1$ for a left turn, $b_i = -1$ for a right turn, and $b_i = 0$ otherwise. The \emph{spirality} of a link $\ell_i$ is defined as $s(\ell_i) = \sum_{j=1}^{i-1} b_j$. Note that, by definition, the spirality of $\ell_1$ is $0$, and by construction the spirality of $\ell_k$ is also $0$.
The \emph{spirality} of a wire is defined as the maximum absolute value of the spirality over all its links.
The spirality of wires in $W_\downarrow$ is defined analogously, going from top to bottom.

\begin{lemma}\label{lem:sameSpirality}
If a wire $w \in W_{\rightarrow}$ and a wire $w' \in W_{\downarrow}$ cross in links $\ell_i$ and $\ell'_j$, then $s(\ell_i)=s(\ell'_j)$.
\end{lemma}
\begin{proof}
By Property (2) $w$ and $w'$ cross exactly once.
Consider an axis-aligned rectangle $R$ that contains the complete drawing and that intersects the first and last link of both $w$ and $w'$.
By definition the spirality of $w$ and $w'$ is zero where they intersect $R$. The wires $w$ and $w'$ subdivide $R$ into four simple faces (see Fig.~\ref{fig:sameSpirality}). Consider the top-left face. Since the face is simple, a counterclockwise tour of the face would increase spirality by four. As $R$ and the intersection of $R$ with $w$ and $w'$ contribute three left turns, the spirality should increase by one when traversing the face from the first link of $w$ to the first link of $w'$, where the spiralities of $\ell_1$ and $\ell'_1$ are $0$. Assuming that the spirality of $\ell_i$ is $x$ and the spirality of $\ell'_j$ is $y$, then we get that $x + 1 - y = 1$ (including the left turn at the crossing). Thus $x = y$.
\end{proof}

\begin{figure}[t]
\begin{minipage}[t]{.55\textwidth}
	\centering
   \includegraphics[]{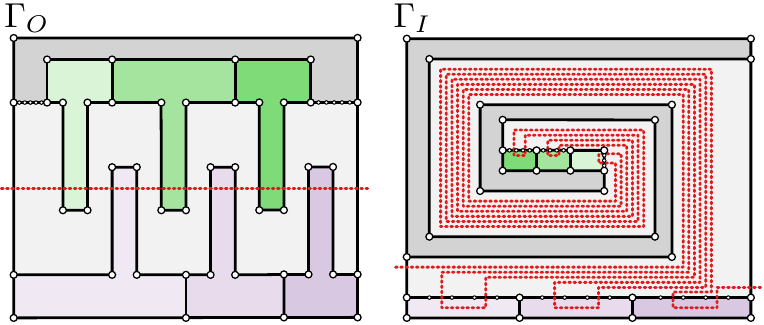}
   \caption{The complexity of a wire can be $\Omega(n^2)$. To satisfy Property (3), the wire in $\Gamma_I$ must spiral through the same polygon a linear number of times. Note that the spirality is still $O(n)$.}
   \label{fig:n2complexity}
\end{minipage}	
\hfill
\begin{minipage}[t]{.42\textwidth}
	\centering
   \includegraphics[]{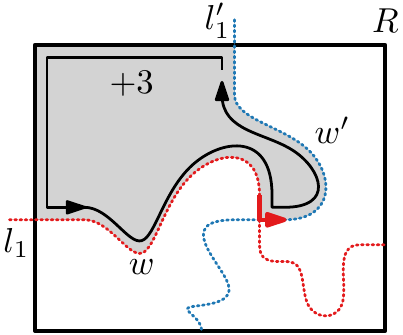}
   \caption{As link $l_1$ and $l'_1$ both have spirality zero and a counterclockwise tour increases spirality by four, the crossing links of $w$ and $w'$ must have equal spirality.}
   \label{fig:sameSpirality}
\end{minipage}
\end{figure}

\subparagraph*{Spirality bound}
In Section~\ref{sec:linear} we show that we can straighten a set of wires using only $O(k)$ linear morphs, if the spirality of each wire is bounded by $k$. It is therefore pertinent to bound the spirality of a proper set of wires. For that we use a particular set of wires. First consider any proper set of wires, which must exist. Then replace every wire $w$ by the lowest and shortest path homotopic to $w$. Because the new wires are homotopic to the initial wires, Property (1) is maintained. Although the wires may partially overlap, they cannot properly cross, and thus overlaps can be removed by slight perturbations. Furthermore, Properties (2) and (3) follow directly from Theorem~\ref{thm:shortestbest} and Lemma~6 from~\cite{Speckmann2018}. Thus, the new set of wires is proper, and in the following we can assume that all wires are lowest and shortest. 

We show that the spirality of any wire in $\Gamma_I$ is $O(n)$. We first establish some properties of the spiralities of links. We consider wires in $W_\rightarrow$ and links with positive spirality, but analogous or symmetric results hold for wires in $W_\downarrow$ and links with negative spirality.

\begin{lemma}\label{lem:intersectSpiral}
Let $\ell_i$ be a horizontal link of a wire $w \in W_\rightarrow$ with even spirality $s(\ell_i) \geq 2$ and let $L$ be a vertical line crossing $\ell_i$. Then there exists a link $\ell_j$ ($j < i$) of $w$ with spirality $s(\ell_i) - 2$ or $s(\ell_i) - 4$ that crosses $L$ below $\ell_i$.
\end{lemma}
\begin{proof}
Let $w[i]$ be the partial wire consisting of links $\ell_1\ldots\ell_i$ of $w$.
We first argue that if $\ell_i$ is the lowest link of $w[i]$ crossing $L$, then $s(\ell_i) \leq 0$. In this case we can create a simple cycle in counterclockwise direction by (1) first going down from $\ell_1$ (starting sufficiently far to the left and going down far enough to avoid crossing the rest of $w[i]$), (2) going to the right until reaching $L$, (3) going up until reaching $\ell_i$, and (4) following $w[i]$ backwards until reaching the starting point of the cycle. The full cycle should be of spirality $4$, and it already contains $3$ left turns by construction. Let $x$ be the contribution of the turn at $\ell_i$ (which can be left or right). Then we get that $3 + x - s(\ell_i) = 4$ or $s(\ell_i) = x - 1$, which directly implies $s(\ell_i) \leq 0$.
Let $\ell_k\in w[i]$ be the link crossing $L$ directly below $\ell_i$ (with $k<i$).
We can again construct a simple cycle consisting of the wire $w[i]$ from $\ell_k$ to $\ell_i$ and the segment of $L$ connecting $\ell_k$ and $\ell_i$ (see Fig~\ref{fig:spiraling}(a)). This cycle contains the bends between $\ell_k$ and $\ell_i$ on $w[i]$, and two bends at the crossings with $L$. Following the cycle in counterclockwise direction implies that there cannot be right turns at both the crossing between $L$ and $\ell_i$ and between $L$ and $\ell_k$, for otherwise $w[k]$ would have to cross $L$ between $\ell_k$ and $\ell_i$ (see Fig.~\ref{fig:spiraling}(b)), contradicting our assumption. Therefore, the bends of $w[i]$ between $\ell_k$ and $\ell_i$ contribute between $2$ and $4$ to the total spirality of $4$ of the cycle (in either direction). As a result, the spirality between $\ell_k$ and $\ell_i$ can differ by at most $4$. We can repeat this argument on $w[k]$. Since the lowest link crossing $L$ has spirality at most $0$, there must exist a link $\ell_j$ crossing $L$ below $\ell_i$ with spirality $s(\ell_i) - 2$ or $s(\ell_i) - 4$.
\end{proof}

\begin{figure}[t]
\begin{minipage}[t]{.51\textwidth}
  \centering
  \includegraphics[]{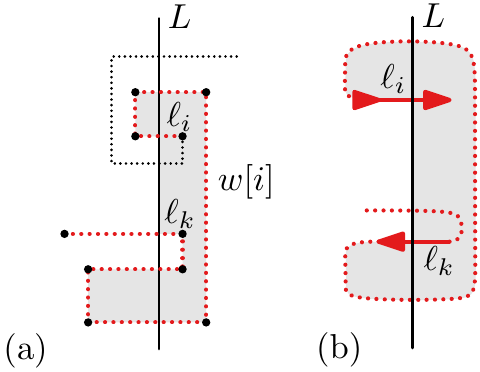}
  \caption{\textbf{(a)} The sub-wire $w[i]$ (red). The path along $w[i]$ from $\ell_k$ to $\ell_i$ and the segment of $L$ connecting $\ell_i$ and $\ell_k$ forms a simple cycle. \textbf{(b)} The cycle cannot have two right turns adjacent to $L$ as $w[i]$ does not intersect $L$ between $\ell_k$ and $\ell_i$.}
  \label{fig:spiraling}
\end{minipage}
\hfill
\begin{minipage}[t]{.46\textwidth}
  \centering
  \includegraphics[]{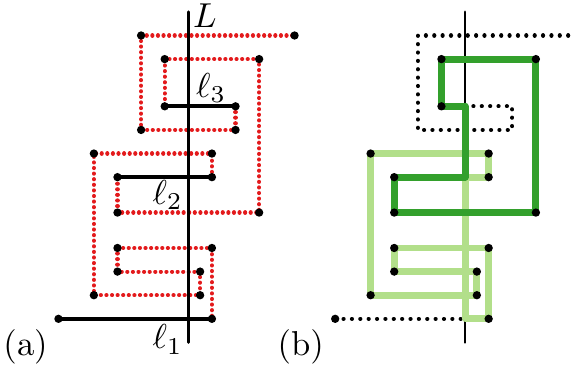}
  \caption{\textbf{(a)} A wire of spirality $s=4$ and $\Omega(s)$ selected edges that are ordered both along $w[i]$ and $L$. \textbf{(b)} A decomposition of $w$ in two polygonal lines and $\Omega(s)$ non-overlapping cycles.}
  \label{fig:spirality}
\end{minipage}		
\end{figure}

If we consider a link $\ell_i$ with negative spirality, then Lemma~\ref{lem:intersectSpiral} holds for a link $\ell_j$ with spirality $s(\ell_i)+2$ or $s(\ell_i)+4$ that intersects $L$ above $\ell_i$.
By repeatedly applying Lemma~\ref{lem:intersectSpiral} we obtain the following result.

\begin{lemma}\label{lem:linearIntersections}
For each link of a wire $w \in W_\rightarrow$ with positive (negative) spirality $s$ there exists a vertical line $L$ and a subsequence of $\Omega(|s|)$ links of $w$ crossing $L$ from bottom to top (top to bottom), such that these links are also ordered increasingly (decreasingly) on spirality.
\end{lemma}

\begin{lemma}\label{lem:spirallinear}
The maximum absolute spirality of any link of a wire $w \in W_\rightarrow$ in $\Gamma_I$ is $O(n)$.
\end{lemma}
\begin{proof}
Without loss of generality assume that the maximum absolute spirality of a link in wire $w$ occurs for a link with positive spirality, and let that spirality be $s$. By Lemma~\ref{lem:linearIntersections} there exists a subsequence $\ell_1, \ldots, \ell_k$ of links of $w$ (not necessarily consecutive along $w$) ordered increasingly by spirality with $k = \Omega(s)$, such that all of these links cross a vertical line $L$ in order from bottom to top. We can thus construct $k-1$ simple cycles by connecting $\ell_i$ to $\ell_{i+1}$ along $L$ and along $w$, such that two different cycles do not share any segments (see Fig.~\ref{fig:spirality}). Since $w$ is constructed to be shortest with respect to its homotopy class, and $G$ is connected, every cycle constructed in this way must cross an edge of $\Gamma_I$, for otherwise $w$ can be shortened by following $L$ locally. By construction $L$ can only cross $O(n)$ edges of $\Gamma_I$, and each edge only once. Similarly, $w$ can cross only $O(n)$ edges of $\Gamma_O$, and each edge only once, as $w$ is horizontal in $\Gamma_O$. As the wires are proper $w$ also crosses only $O(n)$ edges in $\Gamma_I$.
Therefore, the cycles can cross only $O(n)$ edges in total, and thus $s = O(k) = O(n)$.
\end{proof}
Analogously, Lemma~\ref{lem:spirallinear} also holds for wires in $W_\downarrow$, and thus the maximum spirality of all wires is bounded by $O(n)$. Note that we only use shortest paths to bound the spirality of the wires. In the remainder of the paper we merely require that the spirality is bounded by $O(n)$, and any proper set of wires satisfying that bound suffices for our purpose. 

\section{Linear number of linear morphs}\label{sec:linear}

We now describe our algorithm to morph $\Gamma_I$ to $\Gamma_O$ using $O(n)$ linear morphs.
The complexity of the drawing may grow to $O(n^2)$ intermediately though.
In Section~\ref{sec:linearComplexity} we refine our approach to keep the complexity of intermediate drawings at $O(n)$.

It is important to note that for our analysis of the initial spirality we required $\Gamma_I$ and $\Gamma_O$ to be straight-line drawings of the unified graph. For the morph itself we let go of this stringent requirement. During the morph we introduce bends in the edges to rotate them. We will show that the spirality of the wires only decreases during the morph.

The idea of the algorithm is to reduce the maximum spirality of the wires using only $O(1)$ linear morphs. Then, by Lemma~\ref{lem:spirallinear} we need only $O(n)$ linear morphs to straighten the wires, after which we can morph the drawing to $\Gamma_O$ using $O(1)$ linear morphs, as described in Section~\ref{sec:wires}.
We will show that we can reduce the spirality of wires in $W_\rightarrow$ ($W_\downarrow$) without increasing the spirality of wires in $W_\downarrow$ ($W_\rightarrow$) and vice versa. 
In the description below, we limit ourselves to straightening the wires in $W_\rightarrow$.

Now let $\ell^*$ be a link with maximum absolute spirality. To reduce the absolute spirality of $\ell^*$, we use a zigzag-eliminating slide as described in Section~\ref{sec:prelim}, where $\ell^*$ is the middle link of the zigzag. As $\ell^*$ is a link with maximum absolute spirality, the links adjacent to $\ell^*$ are on opposite sides of the line through $\ell^*$. It is easy to see that this slide thus eliminates $\ell^*$ and does not introduce any bends in the wires in $W_\rightarrow$ (see Fig.~\ref{fig:zigzagCombi}(a)). However, the link $\ell^*$ may intersect an edge of $\Gamma_I$ or a link of a wire from $W_\downarrow$. In that case we introduce a bend in the involved edge (link) to execute the slide properly (see Fig.~\ref{fig:zigzagCombi}(b)). If $\ell^*$ intersects more than one edge of $\Gamma_I$, then we must be careful not to introduce an overlap in $\Gamma_I$. To avoid this, we first execute bend-creating slides, essentially subdividing $\ell^*$, to ensure that every link with maximum absolute spirality intersects with at most one edge of $\Gamma_I$ (see Fig.~\ref{fig:zigzagCombi}(c)).

To reduce the number of linear morphs, we combine all slides of the same type into a single linear morph. For all links with the same spirality, all bend-creating slides are combined into one linear morph, and all zigzag-eliminating slides are combined into another linear morph.
Links with positive spirality and links with negative spirality are combined into separate linear morphs.
Thus, using at most $4$ linear morphs, we reduce the 
maximum spirality of all wires in $W_\rightarrow$ by one.

\begin{figure}[t]
  \centering
  \includegraphics[]{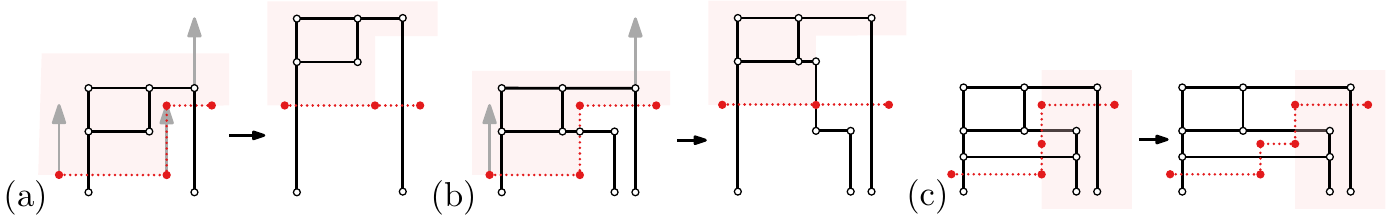}
  \caption{Different slide-types used on the wires. \textbf{(a)} The slide from~\cite{Biedl2013} executed on a wire. \textbf{(b)}~A~single crossing edge (link) causes the introduction of new bends in the edge (link). \textbf{(c)}~A~bend-introducing slide offsets edge intersections without increasing the spirality of the wire.}
  \label{fig:zigzagCombi}
\end{figure}

\subparagraph*{Analysis}
We first show that performing slides on links in $W_\rightarrow$ does not have adverse effects on wires in $W_\downarrow$. This is easy to see for bend-creating slides, as we can assume that wires in $W_\rightarrow$ and wires in $W_\downarrow$ never have overlapping links.
\begin{lemma}\label{lem:independence}
Performing a zigzag-eliminating slide on a link with maximum absolute spirality in $W_\rightarrow$ does not increase the spirality of a wire in $W_\downarrow$.
\end{lemma}

\begin{proof}
Let $\ell^*$ be the middle link of the zigzag-eliminating slide. The zigzag-eliminating slide can only change a wire $w'$ in $W_\downarrow$ if $\ell^*$ crosses a link $\ell'$ in $w'$. By Lemma~\ref{lem:sameSpirality} $s(\ell')=s(\ell^*)$.
The slide does not change the spirality of any link in $w'$, but a new link has been introduced in the middle of $\ell'$. This new link in $w'$ crosses the link obtained by eliminating $\ell^*$, which has absolute spirality $|s(\ell^*)| - 1$. By Lemma~\ref{lem:sameSpirality} the new link in $w'$ must also have absolute spirality $|s(\ell^*)|-1$, and thus the spirality of $w'$ has not been increased.
\end{proof}

We also prove that we can combine zigzag-eliminating slides (and bend-creating slides) into a single linear morph that maintains both planarity and orthogonality of the drawing.

\begin{lemma}\label{lem:merge}
Multiple bend-creating or zigzag-eliminating slides on links of the same spirality in $W_\rightarrow$ can be combined into a single linear morph that maintains planarity and orthogonality.
\end{lemma}
\begin{proof}
As bend-creating slides are simply the inverse of zigzag-eliminating slides, we can restrict ourselves to the latter. As all zigzag-eliminating slides operate on links of the same spirality, they are all either horizontal or vertical. Without loss of generality, assume that all zigzags are horizontal. Then all vertices in the drawing are moved only vertically, which means that the horizontal order of vertices is maintained and that vertical edges remain vertical. Furthermore, since we introduce bends at edges that intersect the middle segment of zigzags, horizontal edges are either subdivided or remain horizontal during the linear morph. Finally, we can only violate planarity if a vertex overtakes an edge in the vertical direction. However, by construction, points with higher $y$-coordinates are moved up at least as far as points with lower $y$-coordinates, and thus the vertical order is also maintained.
\end{proof}

\begin{theorem}\label{thm:linearmorph}
Let $\Gamma_I$ and $\Gamma_O$ be two orthogonal planar drawings of $G$, where $G$ is the unification of $\Gamma_I$ and $\Gamma_O$, and $\Gamma_I$ and $\Gamma_O$ have the same combinatorial embedding and the same outer boundary. Then we can morph $\Gamma_I$ to $\Gamma_O$ using $O(n)$ linear morphs, where $n$ is the number of vertices of $G$.
\end{theorem}
\begin{proof}
Let $W_\rightarrow$ and $W_\downarrow$ be a proper set of wires for $\Gamma_I$ with maximum spirality $O(n)$. As shown in Section~\ref{sec:wires} such a set exists.
Using Lemma~\ref{lem:merge}, we repeatedly reduce the maximum spirality of the wires in $W_\rightarrow$ and $W_\downarrow$ by one using at most two times $4$ linear morphs as described above.
By Lemmata~\ref{lem:spirallinear} and~\ref{lem:independence} all wires can be straightened with at most $O(n)$ linear morphs.
Afterwards, the resulting drawing $\Gamma$ is similar to $\Gamma_O$ except for additional bends. Using $O(1)$ linear morphs we can morph $\Gamma$ to $\Gamma_O$ (see Section~\ref{sec:wires}).
\end{proof}

\section{Linear complexity}\label{sec:linearComplexity}

We refine the approach from Section~\ref{sec:linear} to ensure that the drawing maintains $O(n)$ complexity during the morph.
To achieve this we make two small changes to the algorithm.
First, we ensure that for each edge intersected by links of maximum absolute spirality we only perform a slide for one of the intersecting links and reroute the remaining wires.
This ensures that only $O(1)$ bends per edge are added per iteration of the algorithm. Second, we perform additional intermittent linear morphs to keep the number of bends per edge low.
Both alterations add only $O(n)$ additional linear morphs in total.
The changes ensure that each edge has $O(1)$ bends at every step of the morph; the $O(n)$ complexity bound trivially follows.

\subparagraph*{Rerouting wires}
During each iteration of the algorithm in Section~\ref{sec:linear} we add $O(1)$ linear morphs to ensure that only $O(1)$ new bends are introduced in each edge.
Our approach maintains the invariant that all wires crossing an edge cross the same segment of the edge. Trivially this is the case in $\Gamma_I$. We first establish the following property.

\begin{figure}[t]
\begin{minipage}[t]{.26\textwidth}
  \centering
  \includegraphics[]{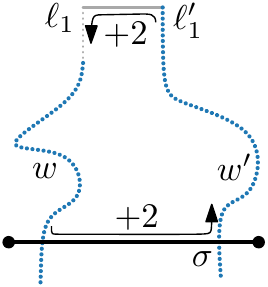}
  \caption{Two wires crossing the same segment $\sigma$ have the same spirality as $s(\ell_1)=s(\ell'_1)=0$ and a counterclockwise tour increases spirality by four.}
  \label{fig:segment}
\end{minipage}
\hfill
\begin{minipage}[t]{.71\textwidth}
  \centering
  \includegraphics[]{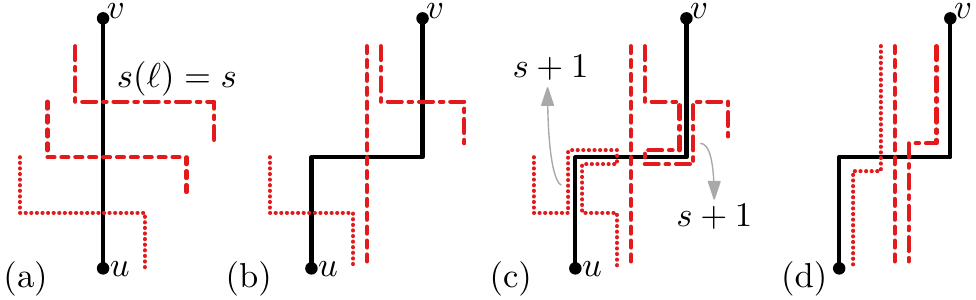}
  \caption{\textbf{(a)} A segment $\sigma=(u,v)$ with three crossing wires of maximum spirality $s>0$. \textbf{(b)} A slide on an arbitrary crossing link adds two bends to $\sigma$. \textbf{(c)} Rerouting the remaining wires ensures all crossing links have spirality $s-1$. New links are created with spirality $s+1$ (and of sprirality $s$) but they do not cross any edge. \textbf{(d)} Result after reducing all non-crossing links of spirality $s$ and $s+1$.}
  \label{fig:reroute}
\end{minipage}
\end{figure}

\begin{lemma}\label{lem:segmentSpirality}
All links intersecting the same segment of an edge have the same spirality.
\end{lemma}
\begin{proof}
Each edge $e$ has a unique orientation in $\Gamma_O$, hence either only wires from $W_\rightarrow$ or $W_\downarrow$ intersect $e$.
All wires cross $e$ in the same direction. Assume without loss of generality that $e$ is horizontal in $\Gamma_O$ and thus only intersected by wires from $W_\downarrow$. Consider two adjacent wires $w,w'\in W_\downarrow$ intersecting segment $\sigma$ of edge $e$.
Consider an additional horizontal segment that would connect link $\ell_1$ of $w$ and $\ell'_1$ of $w'$ (if needed extend $\ell_1$ or $\ell'_1$ upwards).
The area enclosed by wires $w,w'$, the segment $\sigma$, and the extra horizontal segment forms a simple cycle (see Fig.~\ref{fig:segment}). In this cycle there are two left turns at $\sigma$ and two left turns at the additional horizontal segment, and the remaining bends belong to $w$ and $w'$.
If $x$ and $x'$ are the spiralities of $w$ and $w'$ when intersecting $\sigma$, then $x + 2 - x' + 2 = 4$, and thus $x = x'$.
\end{proof}

If multiple links cross a single edge, we execute a slide on only one of these links. We then reroute the remaining wires crossing the edge. This may introduce links with higher absolute spirality, but we can eliminate these links using $O(1)$ linear morphs without affecting the complexity of the drawing. This is formalized in the following lemma.

\begin{lemma}\label{lem:reroute}
The maximum absolute spirality of all links can be reduced while increasing the complexity of each edge by at most $O(1)$.
\end{lemma}
\begin{proof}
Let $s$ be the spirality with the maximum absolute value.
For each edge $e$ crossed by multiple links with spirality $s$, perform a slide for a single crossing link.
By our invariant all links cross $e$ in the same segment $\sigma$ (see Fig.~\ref{fig:reroute}(a)).
Performing a slide on an arbitrary crossing link introduces two new bends in $e$ (see Fig.~\ref{fig:reroute}(b)).
We now reroute the remaining crossing wires in an $\varepsilon$ band along the edge to cross $e$ in the newly created segment (see Fig.~\ref{fig:reroute}(c)). As for a small enough $\varepsilon$ no edge or other wire will be in the area of rerouting, the wires remain a proper set.
The absolute spirality of the crossing links is now $|s|-1$.
The remaining newly created links have absolute spirality at most $|s|+1$.

By Lemma~\ref{lem:merge} $O(1)$ linear morphs are sufficient to perform a slide on all selected crossing links. Similarly $O(1)$ linear morphs are sufficient to remove all links of absolute spirality $|s|+1$ and then all links of absolute spirality $|s|$ (see Fig.~\ref{fig:reroute}(d)). As none of the latter links intersect an edge this does not affect the complexity of the drawing.
\end{proof}

\subparagraph*{Removing excess bends}
Rerouting wires ensures that every edge gathers only $O(1)$ bends when reducing the spirality of all intersecting links by one. However as the maximum absolute spirality, as well as the complexity of $\Gamma_I$, is $O(n)$, the total complexity of the drawing may still become $O(n^2)$ during the morph. We show that using an additional $O(1)$ linear morphs per iteration we can also maintain $O(n)$ complexity.

\begin{lemma}\label{lem:spiraltransfer}
At any point during the morph, the bends with left orientation in an edge are separated from the bends with right orientation by the wires crossing the edge.
\end{lemma}
\begin{proof}
Consider an arbitrary wire $w$ crossing edge $e=(u,v)$.
Let $\ell$ be the link of $w$ that crosses $e$ and assume without loss of generality that $s(\ell)>0$ and that $u$ is on the left side of $w$.
Let $\sigma$ be the segment of $e$ crossed by $\ell$.
We show that when traversing $e$ all right oriented bends occur before the crossing with $w$ and all left oriented bends occur after.
The claim trivially follows.
We consider the orientation of the bends when traversing $e$ from $u$ to $v$.

Clearly the claim already holds in $\Gamma_I$.
Now consider a drawing $\Gamma$ during the morph, where $\ell$ has maximum absolute spirality, and assume the property holds in $\Gamma$.
As $s(\ell)>0$, $\ell$ must be preceded by a left turn and followed by a right turn. Performing a zigzag-eliminating slide on $\ell$ will merge these links into a new link $\ell'$.
A right bend $u'$ is introduced in $\sigma$ left of the intersection with $\ell'$, and thus on the side of $u$, and a left bend $v'$ right of the intersection.
But then when traversing $e$ all right-oriented bends in $e$ occur before the crossing with $w$ and all left-oriented bends occur after.
\end{proof}

We define a \emph{cell} as the area enclosed by two consecutive wires in $W_\rightarrow$ and two consecutive wires in $W_\downarrow$.
By the properties of a proper set of wires, each cell can contain at most one vertex and each edge incident to such a vertex must intersect a different wire.
We now use the following simple approach. Let $\Gamma$ be the drawing after an arbitrary iteration of the algorithm. If a cell in $\Gamma$ contains at least two bends on each edge incident to the vertex of that cell, then we perform $O(1)$ linear morphs to eliminate a bend on each of the incident edges. We can combine the linear morphs for all separate cells.

\begin{figure}[b]
  \centering
  \includegraphics[]{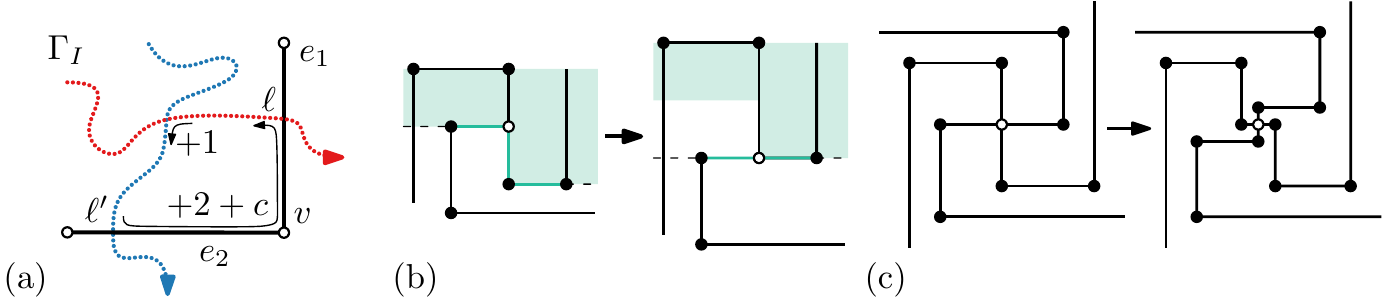}
  \caption{\textbf{(a)} The difference of the spirality of links $\ell$ and $\ell'$ in $\Gamma_I$ is at most two as a counterclockwise tour increases spirality by four. The value of $c$ depends on the actual configuration of $e_1$ and $e_2$ and ranges between $-1$ and $1$. \textbf{(b)} A vertex of degree at most three where all incident edges have at least two bends can be simplified through a series of $O(1)$ zigzag-removing slides. \textbf{(c)} The edges incident to a vertex of degree four can be offset an epsilon amount, after which zigzag-removing slides can reduce the number of bends.}
  \label{fig:spiralityEdges3}
\end{figure}

\begin{lemma}\label{lem:constantBendDifference}
At any point during the morph, inside a single cell for any pair of edges the number of bends differs by at most a constant.
\end{lemma}
\begin{proof}
The statement is vacuously true for cells without a vertex, so consider an arbitrary vertex $v$ and the cell it is contained in.
To prove the statement for $\Gamma$ we first consider the spirality of the links intersecting incident edges of $v$ in $\Gamma_I$.
We show that in $\Gamma_I$ for two incident edges $e_1$, $e_2$ that are adjacent in the cyclic order at $v$ the spirality of the links crossing $e_1$ differs by at most $2$ from those crossing $e_2$.
During the morph we always perform slides on links with maximum absolute spirality and introduce exactly 1 bend inside a cell to reduce the spirality of all links intersecting an edge.
It follows that at any time the difference in the number of bends in incident edges inside the cell is bounded by a constant. 

Edges $e_1$ and $e_2$ have two different possible configurations in $\Gamma_O$. Either one is vertical and the other horizontal, or both are horizontal (vertical).
We consider the case where one is horizontal and the other vertical.
Without loss of generality consider that $e_1$ and $e_2$ are above, respectively, to the left of $v$ in $\Gamma_O$.
By construction $e_1$ and $e_2$ are intersected by a pair of wires $w\in W_\rightarrow$ and $w'\in W_\downarrow$, and they cross before crossing $e_1$ respectively $e_2$.
Wires $w$ and $w'$ together with edges $e_1$ and $e_2$ must then enclose a simple cycle in $\Gamma_O$. As the wires form a proper set this must also be the case in $\Gamma_I$, however, the orientation of the edges may be different in $\Gamma_I$.
The cycle contains three left-corners by construction (see Fig.~\ref{fig:spiralityEdges3}(a)). The turn at $v$ depends on the configuration of $e_1$ and $e_2$ in $\Gamma_I$.
Let $\ell,\ell'$ be the links of $w,w'$ crossing $e_1$ and $e_2$ and let $k$ be the spirality of the links at the crossing of $w$ and $w'$.
We have that $(k-s(\ell))+(s(\ell')-k)+3+c=4$ for $-1\leq c\leq 1$, and thus $|s(\ell')-s(\ell)|\leq 2$.

For the case where both edges are horizontal (vertical) in $\Gamma_O$ a similar argument holds, but now the cycle is formed by two wires from $W_\rightarrow$ and one wire from $W_\downarrow$ resulting in one more left turn. We obtain $s(\ell')-s(\ell)+4+c=4$ for $-1\leq c\leq 1$, resulting in $|s(\ell')-s(\ell)|\leq 1$.
\end{proof}

\begin{corollary}\label{cor:sameSpirality}
All links crossing incident edges of a degree four vertex have the same spirality in $\Gamma_I$.
\end{corollary}
\begin{proof}
For a vertex of degree four two incident edges adjacent in the cyclic order have a fixed relative configuration. We require $c=1$, resulting in $s(\ell)=s(\ell')$.
\end{proof}

\begin{lemma}\label{lem:spiralityEdges}
If all edges incident to a vertex $v$ have at least two bends inside the same cell, then one bend can be removed from each edge without affecting the rest of the drawing.
\end{lemma}
\begin{proof}
If $v$ has at most three incident edges, then there always exists a series of zigzag-removing slides that, in a cyclic order, removes one bend from each of the incident edges  without affecting the other edges and without losing planarity (see Fig.~\ref{fig:spiralityEdges3}(b)).
So assume that $v$ has four incident edges.

By Corollary~\ref{cor:sameSpirality} the spirality of all intersecting links of incident edges is the same in $\Gamma_I$.
Specifically this the spirality is either positive or negative for all intersecting links.
Using Lemma~\ref{lem:spiraltransfer} this implies that all edges will either form only left turns or only right turns inside this cell during the morph.
Assume without loss of generality that all incident edges have only left turns inside this cell and each has at least two left turns. We simultaneously offset all segments incident to $v$ by an epsilon amount, creating a right bend near $v$ in each edge (see Fig.~\ref{fig:spiralityEdges3}(c)). As we only move the segments an epsilon amount we can safely do so without causing new intersections.
Now every incident edge starts with a right-bend followed by a left-bend. Using zigzag-removing slides we remove the newly introduced bend and one of the left-bends.
As these zigzags do not intersect any edge or wire this does not change the spirality of any wire or increase the complexity of any edge.
We can merge the different moves for all vertices together into $O(1)$ linear morphs.
\end{proof}

As each iteration of the refined algorithm increases the complexity of each edge by at most $2$ bends, it is sufficient to reduce complexity of the edges once per iteration.
By Lemma~\ref{lem:spiralityEdges} we can simultaneously simplify all cells where all edges have at least two bends using $O(1)$ linear morphs.
And by Lemma~\ref{lem:constantBendDifference} this requirement is already met when cells contain only $O(1)$ bends.
Cells that do not contain a vertex also do not contain bends as all wires intersect in the same segment of an edge.
It directly follows that the complexity of the drawing is $O(n)$ at all times.
Furthermore, we still need only a linear number of linear morphs.

\begin{theorem}\label{thm:linearcomplex}
Let $\Gamma_I$ and $\Gamma_O$ be two orthogonal planar drawings of $G$, where $G$ is the unification of $\Gamma_I$ and $\Gamma_O$, and $\Gamma_I$ and $\Gamma_O$ have the same combinatorial embedding and the same outer boundary. Then we can morph $\Gamma_I$ to $\Gamma_O$ using $O(n)$ linear morphs while maintaining $O(n)$ complexity during the morph, where $n$ is the number of vertices of $G$.
\end{theorem}

\section{Conclusion}
We described an algorithm that morphs between two planar orthogonal drawings of a connected graph $G$ using only $O(n)$ linear morphs while maintaining planarity and linear complexity of the drawing during the complete morph.
This answers the open question from Biedl~\etal~\cite{Biedl2013}.
As $\Omega(n)$ linear morphs are needed in the worst case, our algorithm is optimal for connected graphs.

Our current proofs only hold for connected graphs.
Specifically Lemma~\ref{lem:spirallinear} assumes that the graph is connected to argue that each cycle must intersect an edge.
By combining the results of Aloupis~\etal~\cite{Aloupis2015} with our work we also obtain an algorithm requiring only $O(n^{1.5})$ linear morphs for disconnected graphs, which still improves on the $O(n^2)$ bound of~\cite{Biedl2013}.
For future work we will investigate if the proofs can be changed to include disconnected graphs.

\newpage

\bibliography{bibliography.bib}
\clearpage

\renewcommand{\textfraction}{0.00}
\renewcommand{\bottomfraction}{1.0}
\renewcommand{\topfraction}{1.0}
\renewcommand{\floatpagefraction}{.9}

\end{document}